\documentclass[aps,preprint]{revtex4}

\usepackage{amsfonts,amsmath,amssymb}
\usepackage{amsthm}
\usepackage{graphicx}
\usepackage{bm,bbm}
\usepackage[utf8]{inputenc}
\usepackage{url}

\newtheorem{theorem}{Theorem}[section]
\newtheorem{definition}{Definition}[section]

\begin{document}

\title{Lie Group S-Expansions and Infinite-dimensional Lie algebras}
\author{Hern\'an Astudillo$^{1}$}
\email{hastudil@udec.cl}
\author{Ricardo Caroca$^{1,2}$}
\email{rcaroca@ucsc.cl}
\author{Alfredo P\'erez$^{1,3}$}
\email{perez@aei.mpg.de}
\author{Patricio Salgado$^{1}$}
\email{pasalgad@udec.cl}

\bigskip

\date{\today }

\begin{abstract}
The expansion method of Lie algebras by a semigroup or S-expansion is
generalized to act directly on the group manifold, and not only at the level
of its Lie algebra. The consistency of this generalization with the dual
formulation of the S-expansion is also verified. This is used to show that the
Lie algebras of smooth mappings of some manifold $M$ \ onto a
finite-dimensional Lie algebra, such as the so called loop algebras, can be
interpreted as a particular kind of $S$-expanded Lie algebras. We consider as
an example the construction of a Yang-Mills theory for an infinite-dimensional
algebra, namely loop algebra.

\end{abstract}

\affiliation{$^{1}$Departamento de F\'{\i}sica, Universidad de Concepci\'{o}n, Casilla 160-C, Concepci\'{o}n, Chile.\\
$^{2}$Departamento de Matem\'{a}tica y F\'{\i}sica Aplicadas, Universidad\\Cat\'{o}lica de la Sant\'{\i}sima Concepci\'{o}n, Alonso de Rivera 2850, Concepci\'{o}n, Chile. \\
$^{3}$ Max-Planck-Institut f\"{u}r Gravitationsphysik, Albert-Einstein-Institut, Am M\"{u}hlenberg 1, D-14476 Golm, Germany.
}
\maketitle

\section{\textbf{Introduction}}

Contractions of finite dimensional Lie algebras were introduced several
decades ago by In\"{o}n\"{u} and Wigner and successfully applied to recover
the Galilei group from the Lorentz group. Subsequently, group contractions
were used to retrieve the Poincar\'{e} group from the de-Sitter group in
various dimensions.

In Ref. \cite{majundar} were discussed the In\"{o}n\"{u}-Wigner contractions
of affine Kac-Moody algebras and in Ref. \cite{fialowski} were studied
contractions of some infinite-dimensional Lie algebras such as Kac-Moody,
Witt, Virasoro and Krichever-Novikov algebras

Expansions of finite dimensional Lie algebras were introduced in Refs.
\cite{hatsuda}, \cite{azcarr}, \cite{azc2}, \cite{irs}, and successfully applied to recover
the $M$-algebra as well as the so called D'Auria-Fr\'{e}-algebra from the
$OSp(32/1)$ Lie algebra. Subsequently, Lie algebras expansions were used to
construct invariant tensor, Chern-Simons and Transgression forms for the
expanded algebras \cite{irs07}, \cite{irs08}, \cite{iprs09-2}.

Recently in Refs. \cite{cms}, \cite{cmps}, the expansion methods were
generalized to obtain expanded higher-order Lie algebras and in Ref. \cite{iprs09-1} was constructed the dual formulation of the Lie algebra
S-expansion procedure whose application permits, for example, to obtain the
$(2+1)$-dimensional Chern-Simons AdS gravity from the so-called
\textquotedblleft exotic gravity\textquotedblright.

It is the purpose of this paper to show that the Lie algebras of smooth
mappings of some manifold $M$ onto a finite-dimensional Lie algebra, such as
the so called loop algebras, can be interpreted as a particular kind of
S-expanded Lie algebras, where the notion of S-expansion is generalized to act
directly on the group manifold and not only at the level of the Lie algebra.
We consider as an example the construction of a Yang-Mills theory for the loop algebra.

\section{\textbf{Lie group Expansions}}

\subsection{\textbf{The S-expansion procedure}}

In this section we shall review some aspects of the S-expansion procedure
introduced in ref. \cite{irs}, \cite{cms}. The S-expansion method is based
on combining the structure constants of the Lie algebra\ $\left(  g,\left[
,\right]  \right)  $ with the inner law of a semigroup $S$ to define the Lie
bracket of a new, S-expanded algebra. Let $S=\left\{  \lambda_{\alpha
}\right\}  $ be a finite Abelian semigroup endowed with a commutative and
associative composition law $S\times S\rightarrow S,$ $\left(  \lambda
_{\alpha},\lambda_{\beta}\right)  \mapsto\lambda_{\alpha}\lambda_{\beta
}=K_{\alpha\beta}^{\text{ \ \ \ \ }\gamma}\lambda_{\gamma}.$ Let the
pair\textbf{\ }$\left(  g,\left[  ,\right]  \right)  $ be a Lie
algebra where $g$\ is a finite dimensional vector space, with basis $\left\{
\mathbf{T}_{A}\right\}  _{A=1}^{\dim g}$, over the field $K$; and $\left[
,\right]  $\ is a ruler of composition $g\times g\longrightarrow g,$ $\left(
\mathbf{T}_{A},\mathbf{T}_{B}\right)  \longrightarrow\left[  \mathbf{T}%
_{A},\mathbf{T}_{B}\right]  =C_{AB}^{\text{ \ \ }C}\mathbf{T}_{C}.$ The direct
product $\mathcal{G}=S\otimes g$ is defined as the Cartesian product
set\textbf{\ }
\begin{equation}
\mathcal{G}=S\otimes g=\left\{  \mathbf{T}_{\left(  A,\alpha\right)  }%
=\lambda_{\alpha}\mathbf{T}_{A}\text{ : }\lambda_{\alpha}\in S\text{ ,
}\mathbf{T}_{A}\in g\right\}  ,\label{s1}%
\end{equation}
endowed with a composition law $\left[  ,\right]  _{S}$ $:\mathcal{G}%
\mathfrak{\times}\mathcal{G}\mathfrak{\rightarrow}\mathcal{G}$ defined by%
\begin{equation}
\left[  \mathbf{T}_{\left(  A,\alpha\right)  },\mathbf{T}_{\left(
B,\beta\right)  }\right]  _{S}=:\lambda_{\alpha}\lambda_{\beta}\left[
\mathbf{T}_{A},\mathbf{T}_{B}\right]  =K_{\alpha\beta}^{\text{ \ \ \ }\gamma
}C_{AB}^{\text{ \ \ }C}\lambda_{\gamma}\mathbf{T}_{C}=C_{\left(
A,\alpha\right)  \left(  B,\beta\right)  }^{\text{ \ \ \ \ \ \ \ \ \ \ }%
\left(  C,\gamma\right)  }\mathbf{T}_{\left(  C,\gamma\right)  },\label{s2'}%
\end{equation}
where $\mathbf{T}_{\left(  A,\alpha\right)  }=\lambda_{\alpha}\mathbf{T}_{A} $
is a basis of $\mathcal{G}\mathfrak{.}$ The set (\ref{s1}) with the
composition law (\ref{s2'}) is called a S-expanded Lie algebra. This algebra
is a Lie algebra structure defined over the vector space obtained by taking
$S$\ copies of $g$
\begin{equation}
\mathcal{G}\mathfrak{:}\oplus_{\alpha\in S}W_{\alpha}\text{ }\left(
W_{\alpha}\approx g\text{, }\mathbf{\forall\alpha}\right)  ,
\end{equation}
$\mathbf{\dim}\mathcal{G}\mathfrak{=}ordS\mathbf{\times\dim}g$\textbf{\ }by
means of the structure constants\textbf{\ }%
\begin{equation}
C_{\left(  A,\alpha\right)  \left(  B,\beta\right)  }^{\text{
\ \ \ \ \ \ \ \ \ \ \ }\left(  C,\gamma\right)  }=C_{AB}^{\text{ \ \ }C}%
\delta_{\alpha\beta}^{\gamma},\label{s2''}%
\end{equation}
where $\delta$ is the Kronecker symbol and the subindex $\alpha,\beta\in S $
denotes the inner composition in $S$ so that $\delta_{\alpha\beta}^{\gamma}=1$
when $\alpha\beta=\gamma$ in $S$ and zero otherwise. The constants $C_{\left(
A,\alpha\right)  \left(  B,\beta\right)  }^{\text{ \ \ \ \ \ \ \ \ \ \ \ }%
\left(  C,\gamma\right)  }$ defined by (\ref{s2''}) inherit the symmetry
properties of $C_{AB}^{\text{ \ \ }C}$ of $g$ by virtue
of the abelian character of the S-product, and satisfy the Jacobi identity.

In a nutshell, the S-expansion method can be seen as the natural
generalization of the In\"{o}n\"{u}-Wigner contraction, where instead of
multiplying the generators by a numerical parameter, we multiply the
generators by the elements of an Abelian semigroup.

\begin{theorem}
The product $\left[  ,\right]  _{S}$ defined in (\ref{s2'}) is also a Lie
product because it is linear, antisymmetric and satisfies the Jacobi identity.
This product defines a new Lie algebra characterized by the pair $\left(
\mathcal{G}\mathfrak{,}\left[  ,\right]  _{S}\right)  $, and is called a
S-expanded Lie algebra.\ 
\end{theorem}

\begin{proof}
The proof is direct and may be found in \cite{irs}, \cite{cms}.
\end{proof}

\subsection{\textbf{S-expansion of Lie groups}}

We can reinterpret the S-expansion procedure as a method to obtain a new Lie
group $\widetilde{G}$ from a given original Lie group $G$, i.e. not only as a
relation between two different Lie algebras but as a method that relates two
different Lie groups\footnote{A step in this direction was given in Ref.
\cite{iprs09-1}.}.

\begin{definition}
Let $\gamma$ be an element of the original Lie group $G$ parametrized as
$\gamma=\exp\left[  \theta^{A}\mathbf{T}_{A}\right]  $, where $\left\{
\theta^{A}\right\}  $, $A=1,..,\dim G$ are the group coordinates and $\left\{
\mathbf{T}_{A}\right\}  $ are the generators of associated Lie algebra $g$
with commutation relations $\left[  \mathbf{T}_{A},\mathbf{T}_{B}\right]
=C_{AB}^{\text{ \ \ }C}\mathbf{T}_{C}$, and let $S=\left\{  \lambda_{\alpha
}\right\}  _{\alpha=1}^{N+1}$ be an abelian semigroup.
\end{definition}

\begin{definition}
The \textbf{S-mapping}\textit{\ }is the mapping from the semigroup $S$ to the
Lie group $G$ ,\ \ $\lambda_{\alpha}\longmapsto\gamma(\lambda_{\alpha})\in G$
defined by%
\begin{equation}
\theta^{A}(\lambda)=\sum_{\alpha=1}^{N+1}\theta^{\left(  A,\alpha\right)
}\lambda_{\alpha},\label{smap}%
\end{equation}
and the parameters $\left\{  \theta^{\left(  A,\alpha\right)  }\right\}  $ are
called the \textbf{S-mapping parameters.}
\end{definition}

The S-mapping is completely characterized by the parameters $\left\{
\theta^{\left(  A,\alpha\right)  }\right\}  $ which, as we will show in the
following theorem, are the group coordinates of the expanded Lie group
$\tilde{G}$.

\begin{theorem}
The S-mapping parameters $\left\{  \theta^{\left(  A,\alpha\right)  }\right\}
$ define the coordinates of a new Lie group $\tilde{G}$, called the S-expanded
Lie group, whose associated Lie algebra is the S-expanded Lie algebra
$\mathcal{G}=S\otimes g$. The elements of the S-expanded Lie group are given
by $\gamma\left(  \lambda\right)  =\exp\left[  \theta^{A}\left(
\gamma\right)  \mathbf{T}_{A}\right]  .$
\end{theorem}

\begin{proof}
Replacing the expression for the S-mapping (\ref{smap}) in the parametrization
of a group element we obtain%
\begin{align}
\gamma\left(  \lambda\right)   &  =\exp\left[  \theta^{A}\left(
\gamma\right)  \mathbf{T}_{A}\right]  ,\nonumber\\
&  =\exp\left[  \sum_{\alpha}\theta^{\left(  A,\alpha\right)  }\lambda
_{\alpha}\mathbf{T}_{A}\right]  ,
\end{align}
by defining $\mathbf{T}_{\left(  A,\alpha\right)  }^{\text{ }}\equiv
\lambda_{\alpha}\mathbf{T}_{A}$ we have%
\begin{equation}
\gamma\left(  \lambda\right)  =\exp\left[  \sum_{\alpha}\theta^{\left(
A,\alpha\right)  }\mathbf{T}_{\left(  A,\alpha\right)  }^{\text{ }}\right]
.\label{smap2}%
\end{equation}

The commutation relations of the $\left\{  \mathbf{T}_{\left(  A,\alpha
\right)  }^{\text{ }}\right\}  $ are given by%
\begin{equation}
\left[  \mathbf{T}_{\left(  A,\alpha\right)  }^{\text{ }},\mathbf{T}_{\left(
B,\beta\right)  }^{\text{ }}\right]  =\left[  \lambda_{\alpha}\mathbf{T}%
_{A},\lambda_{\beta}\mathbf{T}_{B}\right]  =\lambda_{\alpha}\lambda_{\beta
}\left[  \mathbf{T}_{A},\mathbf{T}_{B}\right]  =\lambda_{\alpha}\lambda
_{\beta}\left(  C_{AB}^{\text{ \ \ }C}\mathbf{T}_{C}\right)  ,
\end{equation}
since
\begin{equation}
\lambda_{\alpha}\lambda_{\beta}=K_{\alpha\beta}^{\text{ \ \ }\gamma}%
\lambda_{\gamma},
\end{equation}
we have%
\begin{equation}
\left[  \mathbf{T}_{\left(  A,\alpha\right)  }^{\text{ }},\mathbf{T}_{\left(
B,\beta\right)  }^{\text{ }}\right]  =K_{\alpha\beta}^{\text{ \ \ }\gamma
}\lambda_{\gamma}\left(  C_{AB}^{\text{ \ \ }C}\mathbf{T}_{C}\right)
=K_{\alpha\beta}^{\text{ \ \ }\gamma}\left(  C_{AB}^{\text{ \ \ }C}%
\lambda_{\gamma}\mathbf{T}_{C}\right)  ,
\end{equation}%
\begin{equation}
\left[  \mathbf{T}_{\left(  A,\alpha\right)  }^{\text{ }},\mathbf{T}_{\left(
B,\beta\right)  }^{\text{ }}\right]  =K_{\alpha\beta}^{\text{ \ \ }\gamma
}C_{AB}^{\text{ \ \ }C}\mathbf{T}_{\left(  C,\gamma\right)  }^{\text{ }}.
\end{equation}
The generators $\mathbf{T}_{\left(  A,\alpha\right)  }^{\text{ }}$ obey the
commutation relations of the S-expanded algebra $\mathcal{G}=S\otimes g$, and
according to the eq.(\ref{smap2}) the S-mapping parameters $\left\{
\theta^{\left(  A,\alpha\right)  }\right\}  $ are the group coordinates of the
S-expanded group $\tilde{G}$ with Lie algebra $\mathcal{G}=S\otimes g$. This
completes the proof.
\end{proof}

This theorem establishes the relation between the original Lie group $G$ with
coordinates $\left\{  \theta^{A}\right\}  $ and the S-expanded Lie group
$\tilde{G}$ with coordinates $\left\{  \theta^{\left(  A,\alpha\right)
}\right\}  $ through the eq.(\ref{smap}), where the corresponding Lie algebras
are related through an S-expansion as described in ref. \cite{irs}. In this
context the expansion procedure occurs at the level of the group coordinates.

\subsection{\textbf{Consistency with dual formulation of the S-expansion
procedure}}

In the previous section we have extended the notion of a S-expansion to the
group manifold, where the expansion is done at the level of the group
coordinates. As the Maurer-Cartan forms of the Lie group can be realized as
left-invariant one-forms on the group manifold, the S-expansion at the group
coordinates level should lead automatically to an S-expansion formulation in
terms of the Maurer-Cartan forms. In the next theorem we will prove that this
expansion over the left-invariant forms is consistent with the dual
formulation of the S-expansion procedure of ref. \cite{iprs09-1}.

\begin{theorem}
Let $\left\{  \omega^{A}\left(  \theta^{B}\right)  \right\}  $ be the
Maurer-Cartan forms of the Lie group $G$ and let $\left\{  \omega^{\left(
A,\alpha\right)  }\left(  \theta^{\left(  B,\beta\right)  }\right)  \right\}
$ be the Maurer-Cartan forms of the expanded Lie group $\tilde{G}$ realized as
left invariant one-forms over the corresponding group manifold. The
application of the S-expansion procedure on the group coordinates through an
S-mapping (\ref{smap}) gives the following relation between the Maurer-Cartan
forms of $G$ and $\tilde{G}:$%
\begin{equation}
\omega^{A}=\sum_{\alpha=1}^{N+1}\lambda_{\alpha}\omega^{\left(  A,\alpha
\right)  },
\end{equation}
expression consistent with the dual formulation of the S-expansion procedure
of ref. \cite{iprs09-1}.
\end{theorem}

\begin{proof}
The canonical one form left-invariant on $G$ is given by%
\begin{equation}
\gamma^{-1}d\gamma=e^{-\theta^{A}\mathbf{T}_{A}}de^{\theta^{A}\mathbf{T}_{A}%
}\equiv\omega^{A}\mathbf{T}_{A}.
\end{equation}
Using the identity
\begin{equation}
e^{-\mathbf{X}}de^{\mathbf{X}}=d\mathbf{X}+\frac{1}{2}\left[  d\mathbf{X}%
,\mathbf{X}\right]  +\frac{1}{3!}\left[  \left[  d\mathbf{X},\mathbf{X}%
\right]  ,\mathbf{X}\right]  +\frac{1}{4!}\left[  \left[  \left[
d\mathbf{X},\mathbf{X}\right]  ,\mathbf{X}\right]  ,\mathbf{X}\right]  +...,
\end{equation}
we can obtain the left invariant Maurer-Cartan forms on $G$ in terms of the
group coordinates%
\[
\omega^{A}\left(  \theta\right)  =d\theta^{A}+\frac{1}{2}C_{BE_{1}}^{\text{
\ \ \ \ }A}\theta^{E_{1}}d\theta^{B}+\frac{1}{3!}C_{BE_{1}}^{\text{
\ \ \ }D_{1}}C_{D_{1}E_{2}}^{A}\theta^{E_{1}}\theta^{E_{2}}d\theta^{B}%
\]%
\begin{equation}
+\frac{1}{4!}C_{BE_{1}}^{\text{ \ \ \ }D_{1}}C_{D_{1}E_{2}}^{\text{
\ \ \ }D_{2}}C_{D_{2}E_{3}}^{\text{ \ \ \ \ \ }A}\theta^{E_{2}}\theta^{E_{2}%
}\theta^{E_{3}}d\theta^{B}+...
\end{equation}
Introducing eq.(\ref{smap}) into this expression we have%
\begin{align}
\omega^{A} &  =\left(  \sum_{\alpha}\lambda_{\alpha}d\theta^{\left(
A,\alpha\right)  }\right)  +\frac{1}{2}C_{BE_{1}}^{\text{ \ \ \ }A}\left(
\sum_{\varepsilon_{1}}\lambda_{\varepsilon_{1}}\theta^{\left(  E_{1}%
,\varepsilon_{1}\right)  }\right)  \left(  \sum_{\beta}\lambda_{\beta}%
d\theta^{\left(  B,\beta\right)  }\right)  +\nonumber\\
&  +\frac{1}{3!}C_{BE_{1}}^{\text{ \ \ }D_{1}}C_{D_{1}E_{2}}^{\text{
\ \ \ \ \ }A}\left(  \sum_{\varepsilon_{1}}\lambda_{\varepsilon_{1}}%
\theta^{\left(  E_{1},\varepsilon_{1}\right)  }\right)  \left(  \sum
_{\varepsilon_{2}}\lambda_{\varepsilon_{2}}\theta^{\left(  E_{2}%
,\varepsilon_{2}\right)  }\right)  \left(  \sum_{\beta}\lambda_{\beta}%
d\theta^{\left(  B,\beta\right)  }\right)  +\nonumber\\
&  +\frac{1}{4!}C_{BE_{1}}^{\text{ \ }D_{1}}C_{D_{1}E_{2}}^{\text{
\ \ \ }D_{2}}C_{D_{2}E_{3}}^{\text{ \ \ \ \ \ }A}\left(  \sum_{\varepsilon
_{1}}\lambda_{\varepsilon_{1}}\theta^{\left(  E_{1},\varepsilon_{1}\right)
}\right)  \left(  \sum_{\varepsilon_{2}}\lambda_{\varepsilon_{2}}%
\theta^{\left(  E_{2},\varepsilon_{2}\right)  }\right)  \left(  \sum
_{\varepsilon_{3}}\lambda_{\varepsilon_{3}}\theta^{\left(  E_{3}%
,\varepsilon_{3}\right)  }\right)  \left(  \sum_{\beta}\lambda_{\beta}%
d\theta^{\left(  B,\beta\right)  }\right)  +...,\label{s12}%
\end{align}
rearranging the semigroup elements%
\begin{align}
\omega^{A} &  =\sum_{\alpha}\lambda_{\alpha}d\theta^{\left(  A,\alpha\right)
}+\frac{1}{2}C_{BE_{1}}^{\text{ \ \ \ }A}\sum_{\varepsilon_{1},\beta}%
\lambda_{\varepsilon_{1}}\lambda_{\beta}\theta^{\left(  E_{1},\varepsilon
_{1}\right)  }d\theta^{\left(  B,\beta\right)  }+\frac{1}{3!}C_{BE_{1}%
}^{\text{ \ \ }D_{1}}C_{D_{1}E_{2}}^{\text{ \ \ \ \ }A}\sum_{\varepsilon
_{1},\varepsilon_{2},\beta}\lambda_{\varepsilon_{1}}\lambda_{\varepsilon_{2}%
}\lambda_{\beta}\theta^{\left(  E_{1},\varepsilon_{1}\right)  }\theta^{\left(
E_{2},\varepsilon_{2}\right)  }d\theta^{\left(  B,\beta\right)  }\nonumber\\
&  +\frac{1}{4!}C_{BE_{1}}^{\text{ \ \ }D_{1}}C_{D_{1}E_{2}}^{\text{
\ \ \ \ }D_{2}}C_{D_{2}E_{3}}^{\text{ \ \ \ \ \ \ }A}\sum_{\varepsilon
_{1},\varepsilon_{2},\varepsilon_{3},\varepsilon_{4}}\lambda_{\varepsilon_{1}%
}\lambda_{\varepsilon_{2}}\lambda_{\varepsilon_{3}}\lambda_{\beta}%
\theta^{\left(  E_{1},\varepsilon_{1}\right)  }\theta^{\left(  E_{2}%
,\varepsilon_{2}\right)  }\theta^{\left(  E_{3},\varepsilon_{3}\right)
}d\theta^{\left(  B,\beta\right)  }+....\text{ \ .}\label{s13}%
\end{align}
Using the composition law of the semigroup $\lambda_{\alpha}\lambda_{\beta
}=K_{\alpha\beta}^{\text{ \ \ \ }\gamma}\lambda_{\gamma}$ and taken into
account the fact that the two-selectors behave like a Kronecker delta, we can
introduce a sum into $\alpha$ without changing the result:%
\begin{align}
\omega^{A} &  =\sum_{\alpha}\lambda_{\alpha}dg^{\left(  A,\alpha\right)
}+\sum_{\alpha}\lambda_{\alpha}\frac{1}{2}\sum_{\varepsilon_{1},\beta}\left(
C_{BE_{1}}^{\text{ \ \ \ }A}K_{\beta\varepsilon_{1}}^{\text{ \ \ \ \ \ }%
\alpha}\right)  \theta^{\left(  E_{1},\varepsilon_{1}\right)  }d\theta
^{\left(  B,\beta\right)  }+\nonumber\\
&  +\sum_{\alpha}\lambda_{\alpha}\frac{1}{3!}\sum_{\beta,\varepsilon
_{1},\delta_{1},\varepsilon_{2}}\left(  C_{BE_{1}}^{\text{ \ \ }D_{1}}%
K_{\beta\varepsilon_{1}}^{\text{ \ \ \ \ }\delta}\right)  \left(
C_{D_{1}E_{2}}^{\text{ \ \ \ \ }A}K_{\delta\varepsilon_{2}}^{\text{
\ \ \ \ \ }\alpha}\right)  \theta^{\left(  E_{1},\varepsilon_{1}\right)
}\theta^{\left(  E_{2},\varepsilon_{2}\right)  }d\theta^{\left(
B,\beta\right)  }+\nonumber\\
&  +\sum_{\alpha}\lambda_{\alpha}\frac{1}{4!}\sum_{\varepsilon_{1}%
,\varepsilon_{2},\varepsilon_{3},\varepsilon_{4},\beta,\delta_{1},\delta_{2}%
}\left(  C_{BE_{1}}^{\text{ \ }D_{1}}K_{\beta\varepsilon_{1}}^{\text{
\ \ \ }\delta_{1}}\right)  \left(  C_{D_{1}E_{2}}^{\text{ \ \ \ }D_{2}%
}K_{\delta_{1}\varepsilon_{2}}^{\text{ \ \ \ }\delta_{2}}\right)  \left(
C_{D_{2}E_{3}}^{\text{ \ \ \ \ }A}K_{\delta_{2}\varepsilon_{3}}^{\text{
\ \ \ \ }\alpha}\right)  \theta^{\left(  E_{1},\varepsilon_{1}\right)  }%
\theta^{\left(  E_{2},\varepsilon_{2}\right)  }\theta^{\left(  E_{3}%
,\varepsilon_{3}\right)  }d\theta^{\left(  B,\beta\right)  }+..\label{s14}%
\end{align}
Introducing the structure constants of the expanded Lie algebra $C_{\left(
A,\alpha\right)  \left(  B,\beta\right)  }^{\text{ \ \ \ \ \ \ \ \ \ \ \ \ }%
\left(  C,\gamma\right)  }:=C_{AB}^{\text{ \ \ \ }C}K_{\alpha\beta}^{\text{
\ \ \ }\gamma},$ we have%
\begin{align}
\omega^{A} &  =\sum_{\alpha}\lambda_{\alpha}[\sum_{\beta}\delta_{B}^{A}%
\delta_{\beta}^{\alpha}+\frac{1}{2}\sum_{\varepsilon_{1},\beta}C_{(B,\beta
)\left(  E_{1},\varepsilon_{1}\right)  }^{\text{ \ \ \ \ \ \ \ \ \ \ }\left(
A,\alpha\right)  }\theta^{\left(  E_{1},\varepsilon_{1}\right)  }+\frac{1}%
{3!}\sum_{\beta,\varepsilon_{1},\delta_{1}.\varepsilon_{2}}C_{(B,\beta)\left(
E_{1},\varepsilon_{1}\right)  }^{\text{ \ \ \ \ \ \ \ \ }\left(  D_{1}%
,\delta_{1}\right)  }C_{(D_{1},\delta_{1})\left(  E_{2},\varepsilon
_{2}\right)  }^{\text{ \ \ \ \ \ \ \ \ \ \ \ \ \ \ \ }\left(  A,\alpha\right)
}\theta^{\left(  E_{1},\varepsilon_{1}\right)  }\theta^{\left(  E_{2}%
,\varepsilon_{2}\right)  }\nonumber\\
+ &  \frac{1}{4!}\sum_{\varepsilon_{1},\varepsilon_{2},\varepsilon
_{3},\varepsilon_{4},\beta,\delta_{1},\delta_{2}}C_{(B,\beta)\left(
E_{1},\varepsilon_{1}\right)  }^{\text{ \ \ \ \ \ \ \ \ \ \ }\left(
D_{1},\delta_{1}\right)  }C_{(D_{1},\delta_{1})\left(  E_{2},\varepsilon
_{2}\right)  }^{\text{ \ \ \ \ \ \ \ \ \ \ \ \ \ \ \ }\left(  D_{2},\delta
_{2}\right)  }C_{(D_{2},\delta_{2})\left(  E_{3},\varepsilon_{3}\right)
}^{\text{ \ \ \ \ \ \ \ \ \ \ \ \ \ \ \ }\left(  A,\alpha\right)  }%
\theta^{\left(  E_{1},\varepsilon_{1}\right)  }\theta^{\left(  E_{2}%
,\varepsilon_{2}\right)  }\theta^{\left(  E_{3},\varepsilon_{3}\right)
}+.....]d\theta^{\left(  B,\beta\right)  }.\label{s16}%
\end{align}
Finally we can write
\begin{equation}
\omega^{A}=\sum_{\alpha}\lambda_{\alpha}\omega^{\left(  A,\alpha\right)
},\label{s17}%
\end{equation}
where $\omega^{(A,\alpha)}$ are the Maurer-Cartan forms of the S-expanded
algebra. This completes the proof.
\end{proof}

The eq.(\ref{s17}) gives the relation between the Maurer-Cartan forms of the
original and the expanded Lie algebras. This expression is in agreement with
the dual formulation of the S-expansion method of ref. \cite{iprs09-1}.

\section{\textbf{Infinite-dimensional Lie Algebras}}

In this section we shall review some aspects of the Lie algebras of smooth
mapping of a manifold $M$ onto a finite-dimensional Lie algebra. The main
point of this section is to display the differences between the usual Lie
algebras and the algebras for the $1$-sphere, $2$-sphere and $3$-sphere groups
of a compact simple Lie group.

\subsection{\textbf{The Affine Kac-Moody Algebra }}

Let $g$ be the Lie algebra of a finite-dimensional compact simple Lie group,
$G$. \ Let $Map(S^{1},G)$ be the set of smooth mappings from the circle
$S^{1}$ to $G$:
\begin{equation}
\gamma:S^{1}\longmapsto G;\text{ \ \ }z\longmapsto\gamma(z),\label{k1}%
\end{equation}
where $z\in\mathbb{C},$ $\left\vert z\right\vert =1.$ \ Here $\mathbb{C}$ is
the field of complex numbers. The set of $Map(S^{1},G)$ acquires the structure
of a group under the pointwise multiplication $($in $G):$
\begin{equation}
\left(  \gamma_{1}\circ\gamma_{2}\right)  (z)=\gamma_{1}(z)\gamma
_{2}(z).\label{k2}%
\end{equation}
This is the loop group of $G$. It has the natural structure of an
infinite-dimensional Lie group.

To obtain the corresponding Lie algebra \cite{goddard}\textbf{, }denoted by
\textbf{\ }$Map(S^{1},G)$, we begin by choosing a basis $\mathbf{T}^{\text{
}A},1\leq A\leq\dim(g),$ for $g$ with
\begin{equation}
\left[  \mathbf{T}^{A},\mathbf{T}^{B}\right]  =if^{ABC}\mathbf{T}%
^{C},\label{k3}%
\end{equation}
where the structure constants $f^{abc}$ satisfy the Jacobi identity%
\begin{equation}
f^{ABD}f^{DCE}+f^{CAD}f^{DBE}+f^{BCD}f^{DAE}=0,\label{k4}%
\end{equation}
and antisymmetry%
\begin{equation}
f^{ABC}=-f^{BAC}.\label{k5}%
\end{equation}
Moreover, we choose $f^{ABC}$ to be totally antisymmetric, satisfying%
\begin{equation}
f^{ABC}f^{ABD}=\delta^{CD},\label{k6}%
\end{equation}
which can always be done since $G$ is restricted to be simple. The connected
component of the group $Map(S^{1},G)$ consists of elements
\begin{equation}
\gamma(z)=\exp\left(  -i\mathbf{T}^{A}\theta_{A}(z)\right)  .\label{k7}%
\end{equation}

Near the identity they have the form
\begin{equation}
\gamma(z)\approx1-i\mathbf{T}^{A}\theta_{A}(z),\label{k8}%
\end{equation}
where the functions $\theta_{A}(z)$ are defined on the unit circle. They can
be expanded in a Laurent series%
\begin{equation}
\theta_{A}(z)=\sum_{n=-\infty}^{\infty}\theta_{A}^{n}z^{n}.\label{k9}%
\end{equation}
Thus%
\begin{equation}
\gamma(z)\approx1-i\sum_{n=-\infty}^{\infty}\theta_{A}^{n}\mathbf{T}%
_{n}^{\text{ }A}.\label{k10}%
\end{equation}
Identifying $\theta_{A}^{n}$ as the infinitesimal parameters of $Map(S^{1},G)
$ we see that $\mathbf{T}_{n}^{A}\equiv\mathbf{T}^{\text{ }A}z^{n}$ are the
generators of the algebra $Map(S^{1},g).$ Note that (\ref{k9}) may be viewed
as the expansion of $\theta_{A}(z)$ in a Fourier series; that is, in a
complete set of functions on the unit circle. The basic commutation relations
for the generators $\mathbf{T}_{n}^{A}$ may now be derived%
\begin{equation}
\left[  \mathbf{T}_{m}^{A},\mathbf{T}_{n}^{B}\right]  =if^{ABC}\mathbf{T}%
_{m+n}^{C}.\label{k11}%
\end{equation}
On a basis where the $\mathbf{T}^{A}$ are hermitian operators, a consistent
definition of the hermitian adjoint turns out to be%
\begin{equation}
\mathbf{T}_{n}^{A\dagger}=\mathbf{T}_{-n}^{A}.\label{k12}%
\end{equation}

The loop algebra admits a nontrivial central extension. This central extension
is unidimensional and essentialy unique. This is called the Kac-Moody algebra.
The commutation relations of the Kac-Moody algebra (on the chosen basis) are
\begin{equation}
\left[  \mathbf{T}_{m}^{A},\mathbf{T}_{n}^{B}\right]  =if^{ABC}\mathbf{T}%
_{m+n}^{\text{ }C}+\mathbf{k}\delta^{AB}m\delta_{m,-n},\label{k13}%
\end{equation}
\bigskip where $\mathbf{k}$ is the central generator, which satisfies%
\begin{equation}
\left[  \mathbf{k},\mathbf{T}_{m}^{A}\right]  =0,\text{ \ }\forall
\mathbf{T}_{m}^{A}.\label{k14}%
\end{equation}

\subsection{\textbf{The 2-Sphere and 3-Sphere Algebras}}

The most immediate generalization of the loop group algebra (or Kac-Moody
algebra) is the algebra corresponding to the Lie group $Map(M;G)$ of all
smooth mappings from a compact manifold $M$, dim$M>1,$ onto a compact simple
Lie group $G$ (with the group law of pointwise multiplication in $G$). The set
of smooth mappings from a $n$-dimensional compact manifold $M$ onto $G$ will
form an infinite-dimensional Lie group under pointwise multiplication.

We denote by $Map(S^{n};G)$ the group of smooth mapping $\gamma:S^{n}%
\longmapsto G,$ \ \ $x\longmapsto\gamma(x)$ from the $n$-sphere to $G;$ and by
$Map(S^{n};g)$ the corresponding infinite-dimensional Lie algebra. To see how
the structure $Map(S^{n};g)$ is derived for the cases of $2$-Sphere and
$3$-Sphere algebras, let us consider a general element of the connected
component of the group $Map(S^{n};G)$ with $n=2,3$, which consists of
elements
\begin{equation}
\gamma(x)=\exp\left(  -i\mathbf{T}^{A}\theta_{A}(x)\right)  .
\end{equation}
Near the identity they have the form
\begin{equation}
\gamma(x)\approx1-i\mathbf{T}^{A}\theta_{A}(x),
\end{equation}
where the functions $\theta_{A}(x)$ are a set of $\dim(g)$ functions defined
on $S^{n}$. These functions can be expanded in terms of a complete set or
basis $\left\{  F_{I}\right\}  $ of orthogonal functions on $S^{n}.$ For
functions on the $2$-Sphere $S^{2}$ a basis is provided by the spherical
harmonics $Y_{lm}(z,\varphi)$ ($z=\cos\theta$). For the $3$-Sphere case there
is also a complete set of orthogonal functions provided by the Wigner
$D$-functions $D_{mm^{\prime}}^{j}(\alpha,\beta,\gamma)$, where $j$ is a
non-negative integer or half-odd-integer and $m,m^{\prime}$ separately have
the same spectrum in the range from $-j$ to $j$, changing in steps of one.

Expanding the functions $\theta_{A}(x)$ in the basis $\left\{  F_{I}\right\}
$%
\begin{equation}
\theta_{A}(x)=\sum_{I}\theta_{A}^{I}F_{I}(x),
\end{equation}
we can write
\begin{equation}
\gamma(x)\approx1-i\sum_{I}\theta_{A}^{I}\mathbf{T}^{A}F_{I}(x)=1-i\sum
_{I}\theta_{A}^{I}\mathbf{T}_{I}^{A},
\end{equation}
where, for $n=2$, $I\equiv L$ and $\ F_{L}=Y_{L}(z,\varphi)$, with $L$
denoting the ordered pair $(l,m)$; and for $n=3$, $I\equiv J$ and
$\ F_{J}=D_{J}(\alpha,\beta,\gamma)$, with $J$ standing for the triple
$(j,m,m^{\prime}).$ Identifying $\theta_{A}^{I}$ as the infinitesimal
parameters of $Map(S^{n},G)$ we see that $\mathbf{T}_{I}^{\text{ }A}%
\equiv\mathbf{T}^{A}F_{I}$ are the generators of the algebra $Map(S^{n},g).$
This algebra can be written as the product $g\otimes C^{\infty}(S^{n})$, where
$C^{\infty}(S^{n})$ is the set of smooth functions on the $n$-dimensional
sphere, $S^{n}.$ The commutator is specified by
\begin{equation}
\left[  \mathbf{T}_{I_{1}}^{A},\mathbf{T}_{I_{2}}^{B}\right]  =\left[
\mathbf{T}^{A}F_{I_{1}},\mathbf{T}^{B}F_{I_{2}}\right]  =\left[
\mathbf{T}^{A},\mathbf{T}^{B}\right]  \otimes F_{I_{1}}F_{I_{2}},
\end{equation}
where $F_{I_{1}},F_{I_{2}}\in C^{\infty}(S^{n}).$ Is interesting to note that
when the manifold is $S^{1}$, $Map(S^{n},g)$ becomes the loop algebra.

To derive the expression for the commutation relations of the Lie algebra
$Map(S^{2},g)$ in the basis $\left\{  \mathbf{T}_{L}^{A}\right\}  $, let us
first consider the direct product of two spherical harmonics of the same
arguments. They may be expanded in series as
\begin{equation}
Y_{l_{1}m_{1}}Y_{l_{2}m_{2}}=\sum_{l}c(L_{1},L_{2},L)Y_{lm},
\end{equation}
where%
\begin{equation}
c(L_{1},L_{2},L)=\left\langle l_{1}m_{1},l_{2}m_{2}\mid lm\right\rangle
\left\langle l_{1}0,l_{2}0\mid l0\right\rangle \left[  \frac{(2l_{1}%
+1)(2l_{2}+1)}{4\pi(2l+1)}\right]  ^{1/2}.
\end{equation}
Thus, the commutation relations for the generators $\mathbf{T}_{L}^{A}$ may
now be obtained%
\begin{equation}
\left[  \mathbf{T}_{L_{1}}^{A},\mathbf{T}_{L_{2}}^{B}\right]  =\left[
\mathbf{T}^{A},\mathbf{T}^{B}\right]  \otimes Y_{L_{1}}Y_{L_{2}}=\left(
if^{ABC}\mathbf{T}^{C}\right)  c(L_{1},L_{2},L)Y_{L},
\end{equation}%
\begin{equation}
\left[  \mathbf{T}_{L_{1}}^{A},\mathbf{T}_{L_{2}}^{B}\right]  =if^{ABC}%
c(L_{1},L_{2},L)\mathbf{T}^{C}Y_{L},
\end{equation}%
\begin{equation}
\left[  \mathbf{T}_{L_{1}}^{A},\mathbf{T}_{L_{2}}^{B}\right]  =if^{ABC}%
c(L_{1},L_{2},L)\mathbf{T}_{L}^{C},
\end{equation}
where there is a summation over the dummy index $L$, that is, over
$l\geq\left\vert m\right\vert $ since $m$ has the fixed value $m_{1}+m_{2}.$

For the case of the group $Map(S^{3},G)$ we have that $F_{J}\equiv
D_{mm^{\prime}}^{j}(\alpha,\beta,\gamma)$ providing a complete basis of
functions on $S^{3}$, which may be expanded in series as
\begin{equation}
D_{m_{1}m_{1}^{\prime}}^{j_{1}}D_{m_{2}m_{2}^{\prime}}^{j_{2}}=\sum_{j}%
c(J_{1},J_{2},J)D_{mm^{\prime}}^{j},
\end{equation}
where%
\begin{equation}
c(J_{1},J_{2},J)=\left\langle j_{1}m_{1},j_{2}m_{2}\mid jm\right\rangle
\left\langle j_{1}m_{1}^{\prime},j_{2}m_{2}^{\prime}\mid jm^{\prime
}\right\rangle ,
\end{equation}
and the commutation relations for the generators $\mathbf{T}_{J}^{A}$ are
given by%
\begin{equation}
\left[  \mathbf{T}_{J_{1}}^{A},\mathbf{T}_{J_{2}}^{B}\right]  =if^{ABC}%
c(J_{1},J_{2},J)\mathbf{T}_{J}^{C}.
\end{equation}

This algebra admits a nontrivial central extension, which is of the form
\begin{equation}
\left[  \mathbf{T}_{I_{1}}^{A},\mathbf{T}_{I_{2}}^{B}\right]  =if^{ABC}%
c(I_{1},I_{2},I)\mathbf{T}_{I}^{C}+d_{I_{1}I_{2}(a)}^{ab}\mathbf{K}^{a},
\end{equation}
\bigskip where $\mathbf{K}^{a}$ is the central generators, which satisfies%
\begin{equation}
\left[  \mathbf{K}^{a},\mathbf{K}^{b}\right]  =0;\text{ \ }\left[
\mathbf{K}^{a},\mathbf{T}_{I}^{A}\right]  =0\text{\ ,}%
\end{equation}
where $a$ is an index that labels \ the independent central elements.

\section{\textbf{Infinite-dimensional Lie algebras Map(}$\mathbf{S}^{n},g$)
\textbf{as a S-expanded Lie algebra}}

\subsection{\textbf{Loop Algebra as a }S\textbf{-expanded Lie algebra }}

If, in the S-expansion procedure, the finite semigroup is generalized to the
case of an infinite semigroup, then the S-expanded algebra will be an
\textit{infinite-dimensional algebra}. One can this see from the fact that
$\mathbf{T}_{\left(  A,\alpha\right)  }=\lambda_{A}\mathbf{T}_{\alpha}$
constitutes a base for the S-expanded algebra $\mathcal{G}$ and from the fact
that $\alpha$ takes now the values in an infinite set.

A simple example is the case when the semigroup $S$ becomes the infinite set
of the integer numbers $\mathbb{Z}$\ under the addition. The semigroup
elements can be represented by the following subset of the complex numbers
$S=\left\{  z^{n}=\exp(in\varphi)\right\}  _{n=-\infty}^{\infty}$ for an
arbitrary real parameter $\varphi\in\left[  0,2\pi\right]  $ and $n$ an
integer$.$ The semigroup composition law is given by%
\begin{equation}
z^{n}z^{m}=z^{n+m}=\delta_{l}^{m+n}z^{l}.\label{prod1}%
\end{equation}

The S-expansion procedure over the group manifold described in the section II
has an interesting interpretation. According to the theorem 4, the relation
between the coordinates of the original semigroup $\left\{  \theta
^{A}\right\}  $ and the coordinates of the S-expanded semigroup $\left\{
\theta^{\left(  A,n\right)  }\right\}  $ takes the form%
\begin{equation}
\theta^{A}=\sum_{n=-\infty}^{+\infty}\theta^{\left(  A,n\right)  }z^{n}.
\end{equation}

This expression is equivalent to the Laurent expansion of the group
coordinates defined over $S^{1}$ as in eq.(\ref{k2}) which leads to the loop
algebra. This equivalence is due to the fact that the Fourier modes
$z^{n}=\exp(in\varphi)$ obey the product law (\ref{prod1}) of the semigroup
$\mathbb{Z}$. In a few words, the S-expansion over the group manifold with the
semigroup $\mathbb{Z}$ is equivalent to the "compactification" process over
$S^{1}$ which leads to the definition of the loop algebra.

This means that if the pair\textbf{\ }$\left(  g,\left[  ,\right]  \right)
$\textbf{\ }is a Lie algebra whose finite-dimensional vector space has a basis
given by $\left\{  \mathbf{T}_{A}\right\}  _{A=1}^{\dim g}$, which satisfies
the composition law $g\times g\longrightarrow g,$ $\left(  \mathbf{T}%
_{A},\mathbf{T}_{B}\right)  \longrightarrow\left[  \mathbf{T}_{A}%
,\mathbf{T}_{B}\right]  =if_{ABC}\mathbf{T}_{C}$, then the pair $\left(
\mathcal{G},\left[  ,\right]  \right)  $ \ is a Lie algebra whose vector space
has a basis given by $\left\{  \mathbf{T}_{\left(  A,n\right)  }\right\}  $,
which satisfies a composition law \ $\mathcal{G}\times
\mathcal{G\longrightarrow G},$ defined by%
\begin{equation}
\left[  \mathbf{T}_{\left(  A,m\right)  },\mathbf{T}_{\left(  B,n\right)
}\right]  _{S}=z^{m}z^{n}\left[  \mathbf{T}_{A},\mathbf{T}_{B}\right]
=z^{m}z^{n}\left(  if_{ABC}\mathbf{T}_{C}\right)  ,
\end{equation}
where $\mathbf{T}_{\left(  A,m\right)  }=z^{m}\mathbf{T}_{A}$ is a basis of
$\mathcal{G}\mathfrak{.}$ Since $z^{m}z^{n}=z^{m+n}=\delta_{l}^{m+n}z^{l}$ we
can write%

\begin{equation}
\left[  \mathbf{T}_{\left(  A,m\right)  },\mathbf{T}_{\left(  B,n\right)
}\right]  _{S}=if_{ABC}\delta_{l}^{m+n}z^{l}\mathbf{T}_{C}=if_{ABC}\delta
_{l}^{m+n}\mathbf{T}_{\left(  C,l\right)  },
\end{equation}%
\begin{equation}
\left[  \mathbf{T}_{\left(  A,m\right)  },\mathbf{T}_{\left(  B,n\right)
}\right]  _{S}=if_{ABC}\mathbf{T}_{\left(  C,m+n\right)  }.
\end{equation}

Comparing with eq.(\ref{k11}), we see that the above algebra obtained by an
S-expansion with semigroup $\mathbb{Z}$ can be identified with the loop
algebra. \ 

\subsection{\textbf{The 2-Sphere and 3-Sphere algebras as S-expanded Lie
algebras}}

In the case that the finite semigroup is generalized to a complete set or
basis $\left\{  F_{I}\right\}  $ of orthogonal functions on $S^{n}$, we have
that for the $n=2$ case $S=\left\{  F_{I}=Y_{I}(\theta,\varphi)\right\}  $ and
for the $n=3$ case $S=\left\{  F_{I}=D_{I}(\alpha\beta\gamma)\right\}  ; $ so
that the relation between the coordinates of the original semigroup $\left\{
\theta^{A}\right\}  $ and the coordinates of the S-expanded semigroup
$\left\{  \theta^{\left(  A,n\right)  }\right\}  $ according to the theorem 4
takes the form%
\begin{equation}
\theta^{A}=\sum_{n=-\infty}^{+\infty}\theta^{\left(  A,I\right)  }F_{I}.
\end{equation}
As \ in the loop algebra case, we can interpret this expression as a Fourier
expansion in the base defined over the corresponding compact manifold (in this
case $S^{n}$), where the semigroup elements are identified with the elements
of the corresponding base.

\bigskip The generators of the S-expanded algebra are given by $\mathbf{T}%
_{\left(  A,I\right)  }=\mathbf{T}_{A}F_{I}$ which satisfy a composition law
defined by%

\begin{equation}
\left[  \mathbf{T}_{\left(  A,I_{1}\right)  },\mathbf{T}_{\left(
B,I_{2}\right)  }\right]  =\left[  \mathbf{T}_{A}F_{I_{1}},\mathbf{T}%
_{B}F_{I_{2}}\right]  =\left[  \mathbf{T}_{A},\mathbf{T}_{B}\right]  \otimes
F_{I_{1}}F_{I_{2}}.
\end{equation}
This expression defines the n-sphere Lie algebra.

\section{\textbf{Loop algebra and Yang-Mills Theory}}

In this section we consider the construction of actions invariant under the
loop algebra.

\subsection{\textbf{Invariant tensors for loop algebras}}

In ref. \cite{irs}, was show that the S-expansion procedure permits to
obtaining invariant tensors for expanded Lie algebras from invariant tensors
of the original Lie algebra. This allows constructing gauge theories such as
Yang-Mills or Chern-Simons theories invariant under S-expanded Lie algebras. \ \ 

Following ref.\cite{irs} it is direct to see that if $\left\langle
\mathbf{T}_{A_{1}}\cdots\mathbf{T}_{A_{N}}\right\rangle $ is an invariant
tensor for a Lie algebra $g,$ then an invariant tensor for the loop algebra is
given by: \
\begin{equation}
\left\langle \mathbf{T}_{A_{1}}^{n_{1}}\cdots\mathbf{T}_{A_{N}}^{n_{N}%
}\right\rangle =\overset{+\infty}{\underset{m=-\infty}{\sum}}\alpha^{\left(
m\right)  }\delta_{m}^{n_{1}+n_{2}+\cdots+n_{N}}\left\langle \mathbf{T}%
_{A_{1}}\cdots\mathbf{T}_{A_{N}}\right\rangle ,
\end{equation}
where $\alpha^{\left(  m\right)  }$ are arbitrary constants.

\subsection{\textbf{Yang-Mills theory for the loop algebra of }$SU\left(
N\right)  $}

We consider the Lie algebra $g=SU\left(  N\right)  $, whose Killing metric is
given by
\begin{equation}
\left\langle \mathbf{T}_{A}\mathbf{T}_{B}\right\rangle =\frac{1}{2}\delta
_{AB}.
\end{equation}
This invariant tensor permits constructing a Yang-Mills action given by%

\begin{equation}
S_{YM}=-\frac{1}{4}\int F_{\mu\nu}^{A}F^{\mu\nu B}\left\langle \mathbf{T}%
_{A}\mathbf{T}_{B}\right\rangle .
\end{equation}

If we consider the S-expansion of the Lie algebra $SU\left(  N\right)  $, we
find that the associated Loop algebra is
\begin{equation}
\left[  \mathbf{T}_{A}^{n},\mathbf{T}_{B}^{m}\right]  =if_{ABC}\mathbf{T}%
_{C}^{n+m},\label{eqdos}%
\end{equation}
where $f_{ABC}$ are the structure constants of $SU\left(  N\right)  $. \ The
corresponding invariant tensor for the Loop algebra can be directly obtained:
\begin{equation}
\left\langle \mathbf{T}_{A}^{n}\mathbf{T}_{B}^{m}\right\rangle =\frac{1}%
{2}\overset{+\infty}{\underset{r=-\infty}{\sum}}\alpha^{\left(  r\right)
}\delta_{r}^{n+m}\delta_{AB},\label{eqtres}%
\end{equation}
where $\alpha^{(r)}$ are arbitrary constants. This invariant tensor permits
constructing Yang-Mills actions for Loop algebras. In fact, the Yang-Mills
action for the loop algebra obtained from the Lie algebra $SU\left(  N\right)
$ is%
\begin{equation}
S=-\frac{1}{4}\int d^{d}x\left\langle \mathbf{F}_{\mu\nu}\mathbf{F}^{\mu\nu
}\right\rangle ,\label{eq:4}%
\end{equation}
where $\left\langle \cdots\right\rangle $ corresponds to the invariant tensor
defined in (\ref{eqtres}).

The loop-algebra valued curvature $\mathbf{F}_{\mu\nu}=\overset{+\infty
}{\underset{n=-\infty}{\sum}}F_{\mu\nu\left(  n\right)  }^{A}\mathbf{T}%
_{A}^{n}$ is given by
\begin{equation}
\mathbf{F}_{\mu\nu}=\partial_{\mu}\mathbf{A}_{\nu}-\partial_{\nu}%
\mathbf{A}_{\mu}+i\left[  \mathbf{A}_{\mu},\mathbf{A}_{\nu}\right]  ,
\end{equation}
where $\mathbf{A}_{\mu}=\overset{+\infty}{\underset{n=-\infty}{\sum}}%
A_{\mu\left(  n\right)  }^{B}\mathbf{T}_{B}^{n}$ is the gauge potential. Thus
the components of the curvature are given by
\begin{equation}
F_{\mu\nu\left(  n\right)  }^{B}=\partial_{\mu}A_{\nu\left(  n\right)  }%
^{B}-\partial_{\nu}A_{\mu\left(  n\right)  }^{B}-f_{CDB}\overset{+\infty
}{\underset{k=-\infty}{\sum}}A_{\mu\left(  n-k\right)  }^{C}A_{\nu\left(
k\right)  }^{D}.
\end{equation}

The corresponding gauge transformations are
\begin{equation}
\delta\mathbf{A}_{\mu}=\partial_{\mu}\mathbf{\Lambda}+i\left[  \mathbf{A}%
_{\mu},\mathbf{\Lambda}\right]  ,
\end{equation}
where $\mathbf{\Lambda}$ is the gauge parameter. In terms of its components we
can write%
\begin{equation}
\delta A_{\mu\left(  n\right)  }^{B}=\partial_{\mu}\Lambda_{\left(  n\right)
}^{B}-f_{CDB}\overset{+\infty}{\underset{k=-\infty}{\sum}}A_{\mu\left(
n-k\right)  }^{C}\Lambda_{\left(  k\right)  }^{D}.
\end{equation}

Introducing the invariant tensor (\ref{eqtres}) into the action (\ref{eq:4}),
we obtain the Yang-Mills action for a loop algebra (\ref{eqdos})
\begin{equation}
S=-\frac{1}{8}\overset{+\infty}{\underset{r,n,=-\infty}{\sum}\alpha^{\left(
r\right)  }}\int d^{d}xF_{\mu\nu\left(  n\right)  }^{A}F_{\left(  r-n\right)
}^{A\mu\nu}.
\end{equation}

The coefficients $\alpha^{\left(  r\right)  }$ are arbitrary. The action will
be invariant independently of the values that the above mentioned coefficients acquire.

\section{\textbf{Comments}}

We have generalized the S-expansion method of Lie algebras to the group
manifold, where the expansion occurs at the level of the group coordinates as
is stated in the theorem 4. The consistency with the dual formulation of the
S-expansion procedure has also been checked. This generalization has been
useful to show that the Lie algebras of smooth mappings of some manifold $M$
onto a finite-dimensional Lie algebra , such as the so called Loop algebras,
can be interpreted as particular kinds of S-expanded Lie algebras.

The main results of this paper are: the interpretation of the
infinite-dimensional Lie-algebra $Map(S^{n},g)$ as a particular kind of
S-expanded Lie algebras, where the notion of the S-expansion is generalized to
include the group manifold. The obtaining of invariant tensors for loop
algebras (considered as expanded algebras) allow constructing gauge theories
such as Yang-Mills or Chern-Simons theories invariant under loop-like
algebras. \ \ 

\bigskip

This work was supported in part by FONDECYT through Grants \#s 1080530 and
1070306 and in part by Direcci\'{o}n de Investigaci\'{o}n, Universidad de
Concepci\'{o}n through Grant \# 208.011.048-1.0. A.P. wishes to thank S.
Theisen for his kind hospitality at the MPI f\"{u}r Gravitationsphysik in
Golm, where part of this work was done. He is also grateful to the German
Academic Exchange Service DAAD and the Comisi\'{o}n Nacional de
Investigaci\'{o}n Cient\'{\i}fica y Tecnol\'{o}gica CONICYT, Chile, for
financial support. R.C was supported by grants from the Comisi\'{o}n Nacional
de Investigaci\'{o}n Cient\'{\i}fica y Tecnol\'{o}gica CONICYT and from the
Universidad de Concepci\'{o}n, Chile.

\end{document}